\documentclass{sig-alternate-per}
\usepackage{array}
\usepackage{mdwmath}
\usepackage{mdwtab}
\usepackage{graphicx}
\usepackage{float}
\usepackage{stfloats}
\usepackage{lipsum}
\usepackage{tabularx}
\usepackage[tight,footnotesize]{subfigure}
\usepackage{caption}
\usepackage{epstopdf}
\usepackage{amsmath}
\usepackage{amssymb}
\usepackage{paralist}
\usepackage{mathrsfs}
\usepackage{url}
\usepackage{mathbbol}
\usepackage{times}

\usepackage{booktabs}
\usepackage{verbatim}

\usepackage{bm}			% bold math even for greek letters
\usepackage{cite}

\usepackage{amsthm}
\newtheorem{theorem}{Theorem}%[section]
\newtheorem{definition}{Definition}

\newtheorem{corollary}{Corollary}
\newtheorem{assumption}{Assumption}
\newtheorem{result}{Result}

\usepackage[table]{xcolor}
\usepackage{color,colortbl}

\definecolor{Gray}{gray}{0.9}

\newcommand{\eq}[1]{Eq.~\eqref{#1}}

\newcommand{\bgp}{\texttt{bgp-degree }}
\newcommand{\Dix}{D(i|x)}
 \usepackage[normalem]{ulem}
\begin{document}

%\title{Analysing the Effects of Inter-domain SDN \\on BGP Convergence Time}
\title{Analysing the Effects of Routing Centralization on BGP Convergence Time}
%\subtitle{[Extended Version with Supplementary Material]}
\author{Pavlos Sermpezis\\FORTH, Greece\\sermpezis@ics.forth.gr \and
		Xenofontas Dimitropoulos\\FORTH, Greece\\fontas@ics.forth.gr }

%\numberofauthors{3}
%\author{
%	\begin{tabular*}{\textwidth}%
%		{@{\extracolsep{\fill}}ccc}
%		Pavlos Sermpezis & George Nomikos & Xenofontas Dimitropoulos\\
%		\affaddr{FORTH, Greece} & \affaddr{FORTH, Greece} & \affaddr{FORTH, Greece}\\
%		\email{sermpezis@ics.forth.gr} & \email{gnomikos@ics.forth.gr} & \email{fontas@ics.forth.gr}
%	\end{tabular}*}\\
%}

\maketitle

%\begin{abstract}
%\end{abstract}

\section{Introduction}\label{sec:intro}
Software-defined networking (SDN) has improved the routing functionality in networks like data centers or WANs. Recently, several studies proposed to apply the SDN principles in the Internet's inter-domain routing as well~\cite{Kotronis-CXP-SOSR-2016,Gupta-SDX-CCR-2014,Kotronis-Routing-Centralization-ComNets-2015,Rothenberg-Revisiting-RCP-HotSDN-2012,Caesar-RCP-NSDI-2005}. This could offer new routing opportunities~\cite{Kotronis-CXP-SOSR-2016,Gupta-SDX-CCR-2014} and improve the performance of BGP~\cite{Kotronis-Routing-Centralization-ComNets-2015,Rothenberg-Revisiting-RCP-HotSDN-2012,Caesar-RCP-NSDI-2005}, which can take minutes to converge to routing changes~\cite{Labozitz-Delayed-convergence-CCR-2000,Kushman-Can-Hear-CCR-2007,Oliveira-Quantifying-Path-Exploration-ToN-2009}.
%. A main problem of BGP is its slow convergence~\cite{Labozitz-Delayed-convergence-CCR-2000,Kushman-Can-Hear-CCR-2007,Oliveira-Quantifying-Path-Exploration-ToN-2009}, which can cause severe packet losses~\cite{Labozitz-Delayed-convergence-CCR-2000} and performance degradation in services~\cite{Kushman-Can-Hear-CCR-2007}. 

Previous works have demonstrated that centralization can benefit the functionality of BGP, and improve its slow convergence that causes severe packet losses~\cite{Labozitz-Delayed-convergence-CCR-2000} and performance degradation~\cite{Kushman-Can-Hear-CCR-2007}. However, due to (a) the fact that previous works mainly focus on system design aspects, and (b) the lack of real deployments, it is not clearly understood yet \textit{to what extent} inter-domain SDN can improve performance. %E.g., ``how much can centralization decrease the BGP convergence time?'', or ``how performance relates to SDN penetration?''.

To this end, in this work, we make the first effort towards \textit{analytically} studying the effects of routing centralization on the performance of inter-domain routing, and, in particular, the convergence time of BGP. Specifically, we propose a Markovian model for inter-domain networks, where a subset of nodes (\textit{domains}) coordinate to centralize their inter-domain routing. We then derive analytic results that quantify the BGP convergence time under various network settings (like, SDN penetration, topology, BGP configuration, etc.). Our analysis and results facilitate the performance evaluation of inter-domain SDN networks, which have been studied (till now) only through simulations/emulations that are known to suffer from high time/resource requirements and limited scalability.

%The novelty of our approach lies on the fact that it (a) quantifies the BGP \textit{convergence time}, rather than focusing on conditions for BGP convergence (i.e., if BGP converges or not) as previous analytic works did, and (b) facilitates the performance evaluation of inter-domain SDN networks, which have been investigated only through simulations/emulations that are known to suffer from high time/resource requirements and limited scalability.

%However, it is not clearly understood yet to what extent inter-domain SDN could improve performance. To this end, in this work, we make the first analytic effort towards studying the effects of routing centralization in the performance of inter-domain routing, and, in particular, the convergence time of BGP. We propose a Markovian model for inter-domain SDN networks, and derive analytic results for the BGP convergence time under various network settings (like, SDN penetration, topology, BGP configuration, etc.).

%BGP convergence, which is slow~\cite{Labozitz-Delayed-convergence-CCR-2000,Kushman-Can-Hear-CCR-2007,Oliveira-Quantifying-Path-Exploration-ToN-2009} and can cause severe packet losses and performance degradation in services~\cite{Labozitz-Delayed-convergence-CCR-2000,Kushman-Can-Hear-CCR-2007}

\section{Model}\label{sec:model}
%\subsection{Network Model}
\textbf{Network Model.} We consider a network (e.g., the Internet) composed of $N$ \textit{domains} or \textit{autonomous systems} (ASes). We represent each AS as a single node, i.e., a single BGP router (similarly to~\cite{Kotronis-Routing-Centralization-ComNets-2015}). Such an abstraction allows to hide the details of the internal structure of ASes, and focus on inter-domain routing.
%In reality, an AS can be a very large network composed of hundreds or thousands switches/routers, and extend over a large area. However, such an abstraction allows to hide the details of the internal structure of ASes, and focus on inter-domain routing.

%When two ASes are connected, we consider that a link exists between the corresponding routers, over which data traffic and BGP messages can be exchanged. The routing policies between connected ASes can be customer to provider (transit), peering, etc.

We assume that $k\in[1,N]$ ASes cooperate in order to centralize their inter-domain routing: there exists a \textit{multi-domain SDN controller}, which is connected to the BGP routers of these $k$ ASes\footnote{This system abstraction can capture the main functionality of most of the previously proposed approaches.}. In the remainder, we refer to the set of the $k$ ASes, as the \textit{SDN cluster}.
%To avoid delving into system-specific issues of the centralization, we assume the following generic system: there exists a \textit{multi-domain SDN controller}, which is connected to the BGP routers of these $k$ ASes\footnote{This system abstraction can capture the functionality of most of the previously proposed approaches~\cite{}.}. In the remainder, we refer to the set of the $k$ ASes, as the \textit{SDN cluster}.

%\subsection{BGP Updates}
\textbf{BGP Updates.} As in the Internet, ASes use BGP to exchange information and establish routing paths. %ASes send BGP messages to their neighbors according to their routing policies (transit, peering, etc.).
When a BGP edge router of an AS receives a BGP update, it (i) calculates the updates (if any) for its BGP routing table, (ii) sends updates to the other BGP edge routers within the same AS (e.g., with iBGP), and (iii) sends updates to the BGP routers of the neighboring ASes. The time needed for this process may vary a lot among different connections since it depends on a number of factors, like the employed technology (hardware/software), routers' configuration (e.g., MRAI timers), intra-domain network, etc. 
%As a result, the time between receiving a BGP update and forwarding it to a neighbor AS, (a) may vary a lot among different ASes, routers, times, etc., and (b) is very difficult to be exactly computed due to the involved complexity and diverse factors.
To this end, in order to be able to analytically study the BGP updates dissemination (given the uncertainty and complexity), we model the time between the reception and forwarding of a BGP update in a probabilistic way. %Specifically, we make the following assumption.

\begin{assumption}\label{assumption:exponential-lambda}
The time between the reception of a BGP update in an AS/router and its forwarding to a neighbor AS/router, is exponentially distributed with rate $\lambda$.
\end{assumption}
Despite the simplicity of the above assumption, our results can capture the behavior of real/emulated networks (see Section~\ref{sec:discussion}).
%The above assumption allows us to study the BGP convergence problem using a Markovian framework. Despite its simplicity, our model can capture well the behavior of real/emulated networks, as we show in Section~\ref{}.
%This approach, which is a first step towards analysing inter-domain SDN, allows us to use a Markovian framework to study the BGP convergence problem. We would like to mention here that more complex models (e.g., non-exponential times, heterogeneous rates) could be also used, however, the complexity of the required analysis increases. Moreover, as we show in Section~\ref{}, despite its simplicity, our model can capture well the behavior of real/emulated networks.

%\subsection{Inter-domain Routing}
\textbf{Inter-domain SDN routing.} Each AS belonging to the SDN cluster informs the SDN controller upon the reception of a BGP update. The SDN controller, which is aware of the topology of the SDN cluster (neighbors, policies, paths, etc.), calculates the changes in the routing paths and installs the updated routes in each router/AS belonging to the SDN cluster. ASes react to updates from the SDN controller, as in regular BGP updates, and, thus, forward them to their (non SDN) neighbors.
%The SDN controller, which is aware of the topology of the SDN cluster (i.e., neighbors, policies, and paths for each AS in the SDN cluster), calculates the changes in the routing paths and installs the updated BGP routes in each router/AS belonging to the SDN cluster. ASes react to updates from the SDN controller, as in regular BGP updates, and, thus, forward them to their (non SDN) neighbors.

Let $T_{sdn}$ be the time needed for an AS to inform the SDN controller and the controller to install the updated routes in every AS in the SDN cluster. This time can be expected to be in the order of few seconds~\cite{Kotronis-Routing-Centralization-ComNets-2015}, and much lower than the BGP updating process (cf., the default value for MRAI timers in Cisco routers is $30sec.$), thus, for simplicity, we assume here that $T_{sdn}=0$.

%This time can be expected to be much lower than the BGP updating process, thus, for simplicity, we assume here that $T_{sdn}=0$.
%This time can be expected to be in the order of few seconds~\cite{}, and much lower than the BGP updating process (cf., the default value for MRAI timers is Cisco routers is $30sec.$), thus, for simplicity, we assume here that $T_{sdn}=0$.

\section{Analysis: BGP Convergence Time}\label{sec:analysis}
%\subsubsection*{Preliminaries}
Let us assume a routing change, e.g., an announcement of a new prefix by an AS, in the network at time $t_{0}=0$. Our goal is to calculate the \textit{BGP convergence time}, i.e., the time needed till all ASes/routers in the network have the final (i.e., shortest, conforming to policies) BGP routes for this prefix.
%Let us assume a routing change in the network at time $t_{0}=0$. Without loss of generality, we consider this change to be an announcement of a new prefix by an AS (the analysis for other cases, like prefix withdrawals, changes in the paths of existing prefixes, link failures, etc., is similar). Our goal is to calculate the \textit{BGP convergence time}, i.e., the time needed till all ASes/routers in the network have the final (i.e., shortest, conforming to policies) BGP routes for this prefix.

To this end, using Assumption~\ref{assumption:exponential-lambda}, we can model the dissemination of the BGP updates in the network with the Markov Chain (MC) of Fig.~\ref{fig:mc-nodes}, where each state corresponds to the number of ASes/routers that have the final BGP routes. At time $t_{0}=0$ the system is at state $0$, while the state $N$ denotes the BGP convergence. When an AS in the SDN cluster receives the BGP update, all the nodes in the SDN cluster are informed (through the controller); thus, we have a transition, e.g., from state $i$ to state $k+i$. The transition rates in the MC, as we discuss in detail later, depend on the network topology.
%To this end, using Assumption~\ref{assumption:exponential-lambda}, we can model the dissemination of the BGP updates in the network with the Markov Chain (MC) of Fig.~\ref{fig:mc-nodes}, where each state corresponds to the number of ASes/routers that have the final BGP routes installed. At time $t_{0}=0$ the system is at state $0$, while the state $N$ denotes the BGP convergence. When an AS in the SDN cluster receives the BGP update, all the nodes in the SDN cluster are informed (through the controller); therefore, if $i$ non-SDN ASes have the updated BGP routes and the next node to receive the information is an AS in the SDN cluster, we have a transition from state $i$ to state $k+i$. The transition rates in the MC, as we discuss in detail later, depend on the network topology.

The Markov Chain of Fig.~\ref{fig:mc-nodes} is transient, and the BGP convergence time is the time needed to move from state $0$ to state $N$.

For notation brevity, in the remainder we use the MC of Fig.~\ref{fig:mc-steps}, which is equivalent to the MC of Fig.~\ref{fig:mc-nodes}. Here, the states represent the \textit{number of transitions} in the MC of Fig.~\ref{fig:mc-nodes}. For example, the state/step $1$ corresponds to the state $1$ or $k$ of the MC of Fig.~\ref{fig:mc-nodes}, while the state/step $i$ corresponds to the state $i$ or $k+i-1$ in the MC of Fig.~\ref{fig:mc-nodes}. The states $0$ are equivalent in both MCs, while the state/step $C$ denotes the BGP convergence, and, thus, corresponds to the state $N$ in the MC of Fig.~\ref{fig:mc-nodes}. 

If we denote with $x$ the step at which -for the first time- an AS in the SDN cluster receives the BGP update, then the transitions rates $\lambda_{i}^{'}$ in the MC of Fig.~\ref{fig:mc-steps} are given by
\begin{equation}
\lambda_{i}^{'} = \left\{
\begin{tabular}{ll}
$\lambda_{i,i+1}+\lambda_{i,i+k}$		&$, i\leq x$\\
$\lambda_{k+i-1, k+i}$					&$, i>x$
\end{tabular}
\right.
\end{equation}

\begin{figure}
\subfigure[Markov Chain (number of nodes)]{\includegraphics[width=\linewidth]{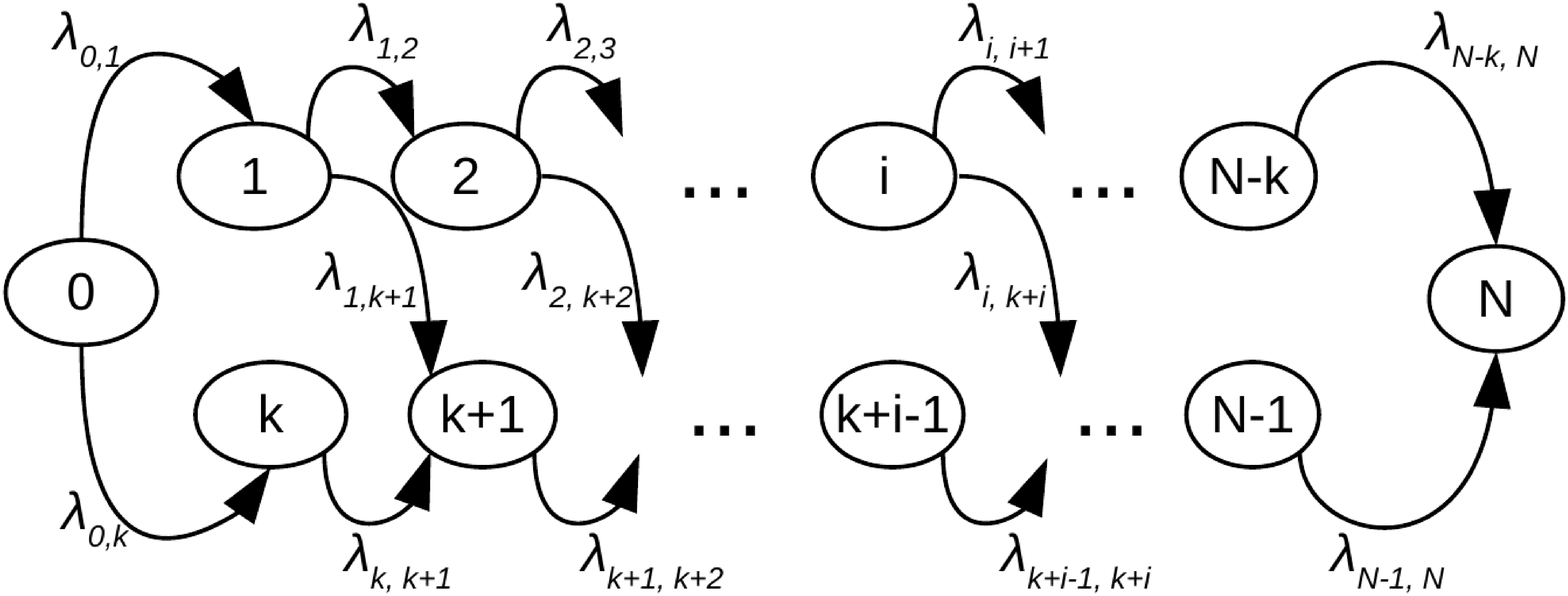}\label{fig:mc-nodes}}
\subfigure[Markov Chain (number of transitions, or \textit{steps})]{\includegraphics[width=\linewidth]{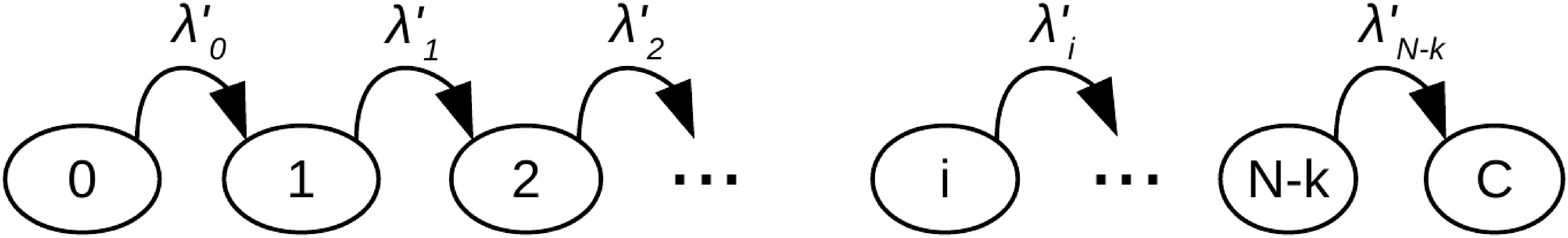}\label{fig:mc-steps}}
\caption{Markov Chains where the states correspond to (a) the number of nodes that have updated BGP routes, and (b) the number of transitions, or \textit{steps}, of the BGP update dissemination process.}
\label{fig:markov-chains}
\end{figure}

%\subsubsection*{Basic Analysis}
We now proceed to calculate the rates $\lambda_{i}^{'}$. 
%Let us consider that the MC is at step $i$, and calculate the rate $\lambda_{i}^{'}$.
The ASes that have received the BGP updates, will then send BGP updates to some of their neighboring ASes, according to their routing policies. We refer to a neighbor to which the update will be forwarded as a \textit{bgp-eligible} neighbor. 

\begin{definition}\label{def:bgp-degree}
We define as the \texttt{bgp-degree} at step $i$, $D(i)$, the number of the ASes that are bgp-eligible neighbors with \underline{any} of the ASes that have received the BGP updates at step $i$. 
\end{definition}

Although an AS might receive the same BGP update from more than one neighbors, the final BGP route will correspond to only one of the received updates (i.e., shortest path). Hence, to calculate the transition rate $\lambda_{i}^{'}$, we take into account only one BGP connection (corresponding to the shortest path) per bgp-eligible neighbor. Since the BGP update times are exponentially distributed with rate $\lambda$ (see, Assumption~\ref{assumption:exponential-lambda}), it follows that $\lambda_{i}^{'}$ will be given by\footnote{The transition time is the minimum of $D(i)$ i.i.d. exponentially distributed times.}
\begin{equation}\label{eq:lambda-prime}
\lambda_{i}^{'} = \lambda\cdot D(i)
\end{equation}

Knowing the rates $\lambda_{i}^{'}$, we can calculate the transition delays in each step. Adding the delays in each step, we can derive Theorem~\ref{thm:expected-delay-generic}, which gives the BGP convergence time, i.e., the time to move from state $0$ to state $C$.

\begin{theorem}\label{thm:expected-delay-generic}
The expectation of the BGP convergence time $T$ in a hybrid SDN/BGP inter-domain topology is given by
\begin{equation}\label{eq:theorem-expected-delay-generic}
E[T] = \frac{1}{\lambda}\cdot \sum_{x=0}^{N-k}\sum_{i=1}^{N-k}\frac{1}{\Dix}\cdot P_{sdn}(x)
\end{equation}
where $\Dix$ is the \bgp of the network at step $i$ given that the SDN cluster receives the update at step $x$, and $P_{sdn}(x)$ is the probability that the SDN cluster receives the update at step $x$.
\end{theorem}
\begin{proof}
To calculate $E[T]$, we first apply the conditional expectation
\begin{equation}\label{eq:ET-sum-ETx}
E[T] = \sum_{x=0}^{N-k}E[T|x]\cdot P_{sdn}(x)
\end{equation}
where $E[T|x]$ denotes the expected convergence time, given that the SDN cluster receives the update at step $x$. Under this condition, the bgp-degrees at each step are $D(i|x)$, $i\in[1,N-k]$. Hence, taking also into account \eq{eq:lambda-prime}, it follows that the transition delay from a step $i$ to a step $i+1$, $T_{i,i+1}$, is exponentially distributed with rate $\lambda_{i}^{'} = \lambda\cdot D(i|x)$, and its expectation is given by
\begin{equation}
E[T_{i,i+1}|x] = \frac{1}{\lambda\cdot D(i|x)}
\end{equation}
\underline{Remark:} The state/step $0$ does not correspond to a real state of the system; it is only used for presentation purposes. Thus, we set $T_{0,1}=0$.

As mentioned earlier, the BGP convergence delay is the time needed to move from step $0$ to step $C$, and thus it is given by the sum of the transition delays of all the intermediate steps, i.e., 
\begin{align}\label{eq:ETx-sum-Dix}
E[T|x] 	&= E\left[\sum_{i=1}^{N-k}T_{i,i+1}|x\right] \nonumber\\
		&= \sum_{i=1}^{N-k}E[T_{i,i+1}|x] \nonumber\\
		&= \sum_{i=1}^{N-k}\frac{1}{\lambda\cdot D(i|x)}
\end{align}
Now, the expression of \eq{eq:theorem-expected-delay-generic} follows by substituting \eq{eq:ETx-sum-Dix} to \eq{eq:ET-sum-ETx}.
\end{proof}

In the following sections we calculate the quantities ${\Dix}$ and $P_{sdn}(x)$ for important network topologies.

\subsection{Full-Mesh Network Topology}
We first consider a basic topology: a full-mesh network, where every AS-pair is connected.

\begin{theorem}\label{thm:P-sdn}
The probability that the SDN cluster receives the update at step $x$ is given by
\begin{equation}\label{eq:P-sdn}
P_{sdn}(x) = \frac{k}{N-x}\cdot \prod_{j=0}^{x-1}\left(1-\frac{k}{N-j}\right)
\end{equation}
\end{theorem}
\begin{proof}
The SDN cluster comprises $k$ (out of the total $N$) ASes. Since we consider that the prefix announcement is made by a (random) AS in the network, the probability that the announcing AS is in the SDN cluster (and thus $x=0$) is 
\begin{equation}
P_{sdn}(0) = \frac{k}{N}
\end{equation}
If the announcing AS is not in the SDN cluster, then $x>0$, and thus
\begin{equation}\label{eq:Psdn-x-larger-0}
P_{sdn}(x>0) = 1-P_{sdn}(0) = 1-\frac{k}{N}
\end{equation}
The probability $P_{sdn}(1)$ is given by
\begin{equation}
P_{sdn}(1) = P_{sdn}(1|x>0)\cdot P_{sdn}(x>0)
\end{equation}
where $P_{sdn}(1|x>0)$ denotes the probability that the SDN cluster receives the BGP update at step $1$, given that it has not received it before. If $x>0$, then at step $1$ the remaining ASes  without the update are $N-1$, of which $k$ belong to the SDN cluster. Since the BGP update processes are distributed with the same rate $\lambda$, the probability that the next AS to get the update belongs to the SDN cluster is $\frac{k}{N-1}$. Therefore, and taking into account \eq{eq:Psdn-x-larger-0}, it holds that 
\begin{equation}
P_{sdn}(1) = P_{sdn}(1|x>0)\cdot P_{sdn}(x>0) = \frac{k}{N-1}\cdot \left(1-\frac{k}{N}\right)
\end{equation}
and, respectively,
\begin{align}
P_{sdn}(x>1) &= \left(1-P_{sdn}(1|x>0)\right)\cdot P_{sdn}(x>0) \nonumber\\
			&= \left(1-\frac{k}{N-1}\right)\cdot \left(1-\frac{k}{N}\right)
\end{align}

Proceeding similarly for the next steps $i=2,...,N-k$, it can be shown that
\begin{equation}
P_{sdn}(i|x>i-1) = \frac{k}{N-i}
\end{equation}
and
\begin{align}
P_{sdn}(x>i-1) 	&= \left(1-P_{sdn}(i-1|x>i-2)\right)\cdot ... \cdot P_{sdn}(x>0) \nonumber\\
				& = \left(1-\frac{k}{N-(i-1)}\right)\cdot...\cdot  \left(1-\frac{k}{N}\right) \nonumber\\
				& = \prod_{j=0}^{i-1}\left(1-\frac{k}{N-j}\right)
\end{align}
and, therefore,
\begin{align}
P_{sdn}(i) &= P_{sdn}(i|x>i-1)\cdot P_{sdn}(x>i-1)\nonumber\\
				&= \frac{k}{N-i}\cdot \prod_{j=0}^{i-1}\left(1-\frac{k}{N-j}\right)
\end{align}
which is the expression of \eq{eq:P-sdn}.
\end{proof}

Theorem~\ref{thm:Dix-full-mesh} gives the bgp-degrees $\Dix$ in a mesh network as a function of $n(i|x)$, which is defined as the number of nodes with updated BGP information at step $i$, given that the SDN cluster received the update at step $x$
\begin{equation}
n(i|x) = \left\{
\begin{tabular}{ll}
$i$	& $, i\leq x$ \\
$i+k-1$	& $, i>x$
\end{tabular}
\right.
\end{equation}

\begin{theorem}\label{thm:Dix-full-mesh}
The \bgp $\Dix$, $i\in[1,N-k], x\in[0,N-k]$, in a full-mesh network topology is given by
\begin{equation}
\Dix = N-n(i|x)
\end{equation}
\end{theorem}
\begin{proof}
In a mesh network, since every AS-pair is directly connected, only the BGP messages sent by the announcing AS (i.e., shortest path) need to be considered. In step $i$, the announcing AS has $N-n(i|x)$ neighbors that have not received the BGP updates, and thus, it follows that $D(i|x) = N-n(i|x)$.
%In a mesh network, since every AS-pair is directly connected, the final BGP routes in each AS correspond to the BGP update it receives from the announcing AS. In other words, only the BGP messages sent by the announcing AS need to be considered. In step $i$, the announcing AS has $N-n(i|x)$ neighbors that have not received the BGP updates, and thus, it follows that $D(i|x) = N-n(i|x)$.
\end{proof}

\subsection{Random Graph Network Topologies}

In networks that are not full-meshes, ASes can be connected in many different ways and policies. Since it is not possible to study every single topology, we use two classes of random graphs to capture the effects of routing centralization in non full-mesh networks. In this first approach, we consider unconstrained routing policies.% (see, Section~\ref{}, for a constrained policies case).

\subsubsection{Poisson (Erdos-Renyi) Graph}
We first consider the case of a Poisson random graph, where a link between each AS-pair exists with probability $p$. Varying the value of $p$ we can capture different levels of sparseness.

Using similar arguments as in the full-mesh case, it is easy to show that the probabilities $P_{sdn}(x)$ are given by Theorem~\ref{thm:P-sdn}. The \textit{expected} bgp-degrees, which can be used (as an approximation) instead of $\Dix$ in Theorem~\ref{thm:expected-delay-generic}, are given by the following Theorem. 
\begin{theorem}\label{thm:Dix-poisson}
The expectation of the \bgp $\Dix$, $i\in[1,N-k], x\in[0,N-k]$, in a Poisson graph network topology is
\begin{equation}
E[\Dix] = \left(N-n(i|x)\right) \cdot \left(1-(1-p)^{n(i|x)}\right)
\end{equation}
\end{theorem}
\begin{proof}
In a non full-mesh network, some ASes are not directly connected to the announcing AS. Thus, in the calculation of the $\Dix$ we need to consider the bgp-eligible neighbors of \textit{all} the ASes that have received the update.

Let assume that we are at step $i$, and $n(i)$ nodes have received the BGP updates; we denote the set of these nodes as $S_{i}$. A node $m\notin S_{i}$ is a bgp-eligible neighbor with a node $j\in S_{i}$ with probability
\begin{equation}
P(m,j) = p
\end{equation}
since every pair of nodes is connected with probability $p$ (by the definition of a Poisson graph). The probability that $m$ is \textit{not} a bgp-eligible neighbor with \textit{any} of the nodes $j\in S_{i}$, is given by
\begin{align}
1-P(m,S_{i}) = \prod_{j\in S_{i}}(1-P(m,j)) = \prod_{j\in S_{i}}(1-p) = (1-p)^{n(i)}
\end{align}
since $|S_{i}| = n(i)$. It follows easily that the complementary event, i.e., $m$ is a bgp-eligible neighbor with \textit{any} of the nodes $j\in S_{i}$, happens with probability
\begin{align}
P(m,S_{i}) = 1- (1-p)^{n(i)}
\end{align}

There are $N-n(i)$ ASes without the update, with each of them being a bgp-eligible neighbor with any of the nodes $j\in S_{i}$ with (equal) probability $P(m,S_{i})$. Hence, the total number of bgp-eligible neighbors (or, as defined in Def.~\ref{def:bgp-degree}, the \textit{bgp-degree} $D(i)$) is a binomially distributed random variable, whose expectation is given by 
\begin{equation}
E[D(i)] = (N-n(i))\cdot (1-(1-p)^{n(i)})
\end{equation}
\end{proof}

\begin{corollary}
Using the expectation of $D(i|x)$ in Theorem~\ref{thm:expected-delay-generic}, underestimates the BGP convergence time, i.e.,
\begin{equation}
E[T] \geq \frac{1}{\lambda}\cdot \sum_{x=0}^{N-k}\sum_{i=1}^{N-k}\frac{1}{E[\Dix]}\cdot P_{sdn}(x)
\end{equation}
\end{corollary}
\begin{proof}
The bgp-degree $\Dix$ in non full-mesh networks, is a random variable that can take different values, depending on the BGP updates dissemination process (i.e., the exact set of nodes that have received the BGP updates at step $i$, and their links to the rest of the nodes). Thus, we can write for the transition delay 
\begin{equation}\label{eq:transition-delay-generic}
E[T_{i,i+1}|x] = \sum_{y} \frac{1}{y} \cdot P\{\Dix = y\} = E\left[\frac{1}{\Dix}\right]
\end{equation}
Calculating the exact value of the expectation $E\left[\frac{1}{\Dix}\right]$ is difficult, thus, we use a well known approximation in \eq{eq:transition-delay-generic}:
\begin{equation}
E[T_{i,i+1}|x] = E\left[\frac{1}{\Dix}\right] \approx \frac{1}{E[\Dix]}
\end{equation}
where the calculation of $E[\Dix]$ is much easier (see, e.g., proof of Theorem~\ref{thm:Dix-poisson}). Then the BGP convergence delay is approximately given by
\begin{equation}
E[T] \approx \frac{1}{\lambda}\cdot \sum_{x=0}^{N-k}\sum_{i=1}^{N-k}\frac{1}{E[\Dix]}\cdot P_{sdn}(x)
\end{equation}
However, applying Jensen's bound for the expectation of a convex function (here, $f(x) = \frac{1}{x}$) of a random variable (here, $\Dix$) on \eq{eq:transition-delay-generic}, gives
\begin{equation}
E[T_{i,i+1}|x] = E\left[\frac{1}{\Dix}\right] \geq \frac{1}{E[\Dix]}
\end{equation}
which proves the Corollary.
\end{proof}

\subsubsection{Arbitrary Degree Sequence Random Graph}
The structure of networks where the degrees (i.e., the number of connections) of the ASes are largely heterogeneous, e.g., power-law graphs, can be better described with a Configuration-Model Random Graph (CM-RG) rather than a Poisson graph. In the CM-RG model, a random graph is created by connecting randomly the nodes (i.e., ASes), whose degrees are given~\cite{Newman:Networks-book}. Hence, we can use the CM-RG to model a network with \textit{any arbitrary degree sequence} with mean value $\mu_{d}$ and variance $\sigma_{d}^{2}$ (and, $CV_{d} = \frac{\sigma_{d}}{\mu_{d}}$).

If the participation of an AS in the SDN cluster is independent of its degree, then the probabilities $P_{sdn}(x)$ are given by Theorem~\ref{thm:P-sdn}, and the bgp-degrees $\Dix$ are given by the following Result\footnote{We use the notation ``Result'', instead of ``Theorem'', because the provided expression is an approximation.}.
\begin{result}
The expectation of the \bgp $\Dix$, $i\in[1,N-k], x\in[0,N-k]$, in a CM-RG network topology is given by
\begin{equation}\label{eq:Dix-CM}
E[\Dix] = D(1|x)\cdot \prod_{j=1}^{i-1}A(j|x) + \sum_{j=1}^{i-1}\left(\mu_{d}(j|x)-1\right)\cdot \prod_{m=j+1}^{i-1}A(m|x)
\end{equation}
where
\begin{align}\label{eq:CM-D1x}
D(1|x)&= \left\{
\begin{tabular}{ll}
$\mu_{d}$	& $, x>0$ \\
$(N-k)\cdot \mu_{d} \cdot \ln\left(\frac{N}{N-k}\right)$	& $, x=0$
\end{tabular}
\right.\\
\mu_{d}(j|x) 	&= \mu_{d}\cdot \prod_{m=1}^{j-1}\left(1-\frac{CV_{d}^{2}}{N-n(m|x)-1}\right) \label{eq:average-out-degree}\\
A(j|x) 			&= 1-\frac{\mu_{d}(j|x)}{N-n(j|x)-1}
\end{align}
\end{result}
\begin{proof}
The main difference with the Poisson case is that in the CM-RG case, it is more probable that the ASes with the higher degrees will receive the BGP updates faster. For instance, let us assume that the announcing AS, e.g., AS-1, does not belong to the SDN cluster. If we denote with $d_{1},d_{2}$, and $d_{3}$ the degrees of AS-1, AS-2 and AS-3 (where AS-2 and AS-3, have not received yet the BGP update), a property of a CM-RG says that AS-1 is directly connected with AS-2 and AS-3 with probabilities
\begin{equation}
P(1,2) = c\cdot d_{1}\cdot d_{2}~~~~and~~~P(1,3) = c\cdot d_{1}\cdot d_{3}
\end{equation}
respectively, where $c$ a normalizing constant. In other words, the AS with with the higher degree has a higher probability to be directly connected to AS-1. Consequently, ASes with higher degrees have a higher probability to get the BGP update faster.

Now, let us first derive \eq{eq:CM-D1x}. If $x>0$, the announcing AS does not belong to the SDN cluster ($n(1|x>0)=1$), and thus the bgp-degree will be equal to the degree of the announcing AS. Since the average degree of a node is $\mu_{d}$, it follows easily that expectation of the bgp-degree in this case, is given by
\begin{equation}
E[D(1|x>0)] = \mu_{d}
\end{equation}

If $x=0$, the announcing AS belongs to the SDN cluster, and, thus, all the $k$ nodes in the SDN cluster have the BGP updates. In this case, the bgp-degree is the number of all bgp-eligible neighbors of these $k$ nodes. Let us denote with $S_{1}$ the set of nodes in the SDN cluster. Since the fact that a node belongs to the SDN cluster and its degree are independent, the probability that an edge coming out from a node $m\notin S_{1}$ is connected to a node $j\in S_{1}$, is equal to $\frac{k}{N}$. Hence, a node $m\notin S_{1}$, with degree $d_{m}$, is not connected to any of the $k$ nodes in the SDN cluster with probability
\begin{equation}
1-P(m,S_{1}) = \left(1-\frac{k}{N}\right)^{d_{m}}
\end{equation}
and, respectively
\begin{equation}
P(m,S_{1}) = 1-\left(1-\frac{k}{N}\right)^{d_{m}}
\end{equation}
The above equation holds $\forall m\notin S_{1}$; the degrees $d_{m}$ can have different values for each $m$. Since, there are $N-k$ nodes that do not belong to the SDN cluster (and, do not have the BGP updated route), the expected bgp-degree is given by
\begin{equation}
E[D(1|0)] = (N-k)\cdot E\left[1-\left(1-\frac{k}{N}\right)^{d}\right]
\end{equation}
where the expectation is taken over $d$, i.e., over all the degrees $d_{m}, m\notin S_{1}$. To calculate this expectation is difficult, thus we approximate it using a Taylor series approximation, i.e, 
\begin{align}
E\left[1-\left(1-\frac{k}{N}\right)^{d}\right] &= 1-E\left[\left(1-\frac{k}{N}\right)^{d}\right]  \nonumber\\
		&\approx 1-\left(1+E[d]\cdot \ln\left(1-\frac{k}{N}\right)\right) \nonumber\\
		&= - E[d]\cdot \ln\left(1-\frac{k}{N}\right)\nonumber\\
		&= E[d]\cdot \ln\left(\frac{N}{N-k}\right) \nonumber \\
		&= \mu_{d}\cdot \ln\left(\frac{N}{N-k}\right)
\end{align}
which completes the derivation of \eq{eq:CM-D1x}.

To compute the bgp-degrees $D(i|x)$ of the steps $i=2,...,N-k$, we follow a methodology similar to~\cite{pavlos-conf-model}. Let $D(i-1)$ be the bgp-degree at step $i-1$ and $\mu_{d}(i-1)$ the average degree of the nodes that have not received the BGP update by step $i-1$  (i.e., the nodes $m, m\notin S_{i-1}$). The average degree $\mu_{d}(i-1)$ is not equal to $\mu_{d}$, in general; this is due to the fact that nodes with higher degrees receive faster the BGP updates and thus the remaining nodes are nodes with lower degrees (for a more detailed argumentation see~\cite{pavlos-conf-model}). Following similar arguments as in~\cite{pavlos-conf-model}, it can be shown that the average degrees $\mu_{d}(i)$ are approximately given by \eq{eq:average-out-degree}.

We calculate the bgp-degree of the step $i$, based on the quantities $D(i-1)$ and $\mu_{d}(i-1)$, as follows
\begin{equation}\label{eq:Di_Di-1}
D(i) = D(i-1) - 1 + \mu_{d}(i-1)\cdot \left(1- \frac{D(i-1)}{N-n(i-1)}\right)
\end{equation}
The term $D(i-1) - 1$ is the bgp-degree of the previous step minus $1$, which denotes the $i^{th}$ node that received the BGP update (i.e., at the transition between step $i-1$ and $i$). To this quantity, we need to add the number of nodes that are bgp-eligible neighbors of the $i^{th}$ node, but are not bgp-eligible neighbors of any of the nodes $\in S_{i-1}$. The total nodes without the updated BGP routes at step $i-1$ are $N-n(i-1)$; $D(i-1)$ of these nodes are bgp-eligible neighbors of at least one node $\in S_{i-1}$. Hence, the probability that a bgp-eligible neighbor of the $i^{th}$ node is not a bgp-eligible neighbor of any of the nodes $\in S_{i-1}$, is given by $\left(1- \frac{D(i-1)}{N-n(i-1)}\right)$. Since, the $i^{th}$ node has (on average) $\mu_{d}(i-1)$ bgp-eligible neighbors, it follows that the number we need to add to the quantity $(D(i-1) - 1)$ is 
\begin{equation}
\mu_{d}(i-1) \cdot \left(1- \frac{D(i-1)}{N-n(i-1)}\right)
\end{equation} 

Finally, calculating recursively \eq{eq:Di_Di-1}, after some algebraic manipulations, we derive \eq{eq:Dix-CM}. 

\end{proof}
\section{Validation and Discussion}\label{sec:discussion}
We built a simulator, conforming to our model (i.e., Assumption~\ref{assumption:exponential-lambda}, $T_{sdn}=0$, etc.). In Fig.~\ref{fig:TvsPenetration} we compare our theoretical results against simulations (averages over $200$ runs). The accuracy is high for the Full-mesh and Poisson graph cases (Fig.~\ref{fig:TvsPenetration-fm-p}). In the CM-RG case (Fig.~\ref{fig:TvsPenetration-cm}) the two curves are similar, and the error of our expression (which is an approximation) is always less than $18\%$. 

Moreover, an initial comparison with the emulation results of~\cite{Kotronis-Routing-Centralization-ComNets-2015}, where a \textit{real BGP software router} is employed, shows that our theory (despite the assumptions and model simplicity) is in agreement with their observations: e.g., the convergence time (i) decreases faster after $\frac{k}{N}>50\%$, (ii) has small differences among different topologies when centralization exists, etc. %(cf.~\cite[Fig.~3]{Kotronis-Routing-Centralization-ComNets-2015}).

This highlights one of the contributions of our work: with our results, we can quickly evaluate the effects of routing centralization. The provided expressions are simple (need only to know a few parameters: $N$ and $k$, and -if needed- $p$, or $\mu_{d}$ and $CV_{d}$) and easy/fast to compute, whereas emulations are time/resource demanding and have limited scalability. Hence, one could use our work to obtain initial insights, and then, e.g., proceed to a few targeted emulations for a more fine-grained or system-specific investigation.

A second contribution is that our methodology and results can be used as the building blocks to model and analyse more complex settings; i.e., model different parts of a large network as full-mesh/Poisson/CM-RG sub-networks, use our expressions for each sub-network, and synthesize them to compute the overall network performance. We consider such an example in Section~\ref{sec:case-study}, where we model the core of the Internet, considering different classes of tier-1 and tier-2 ISPs, while taking into account their routing policies as well.

Finally, we believe that this first analytic study, offers a new useful approach for investigating inter-domain SDN, and can be the basis for analysing further aspects of this field. In particular, we plan to extend our methodology, and consider more settings (e.g., BGP update times), performance metrics, and applications.

\begin{figure}
\subfigure[]{\includegraphics[width=0.49\linewidth]{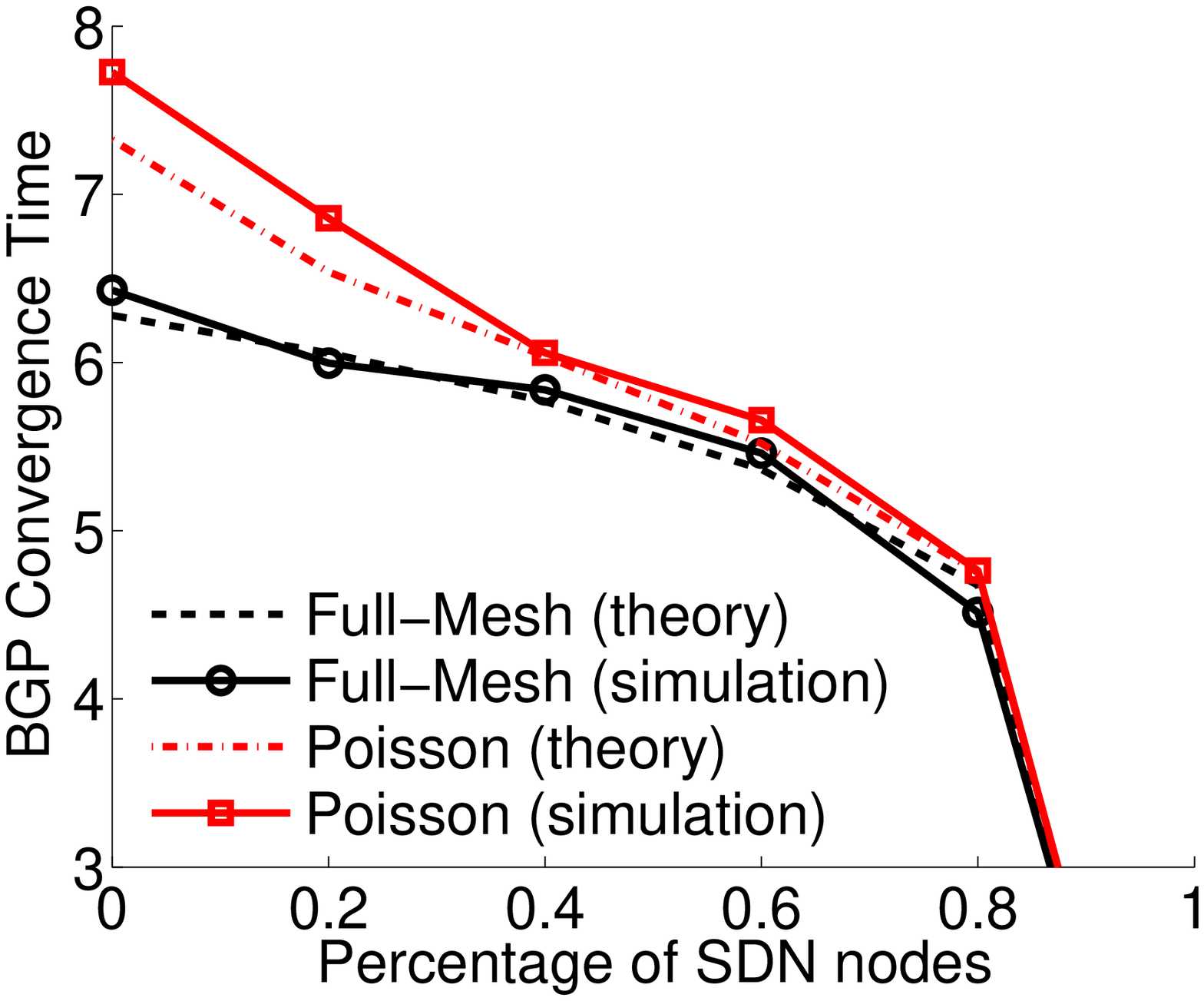}\label{fig:TvsPenetration-fm-p}}
\subfigure[]{\includegraphics[width=0.49\linewidth]{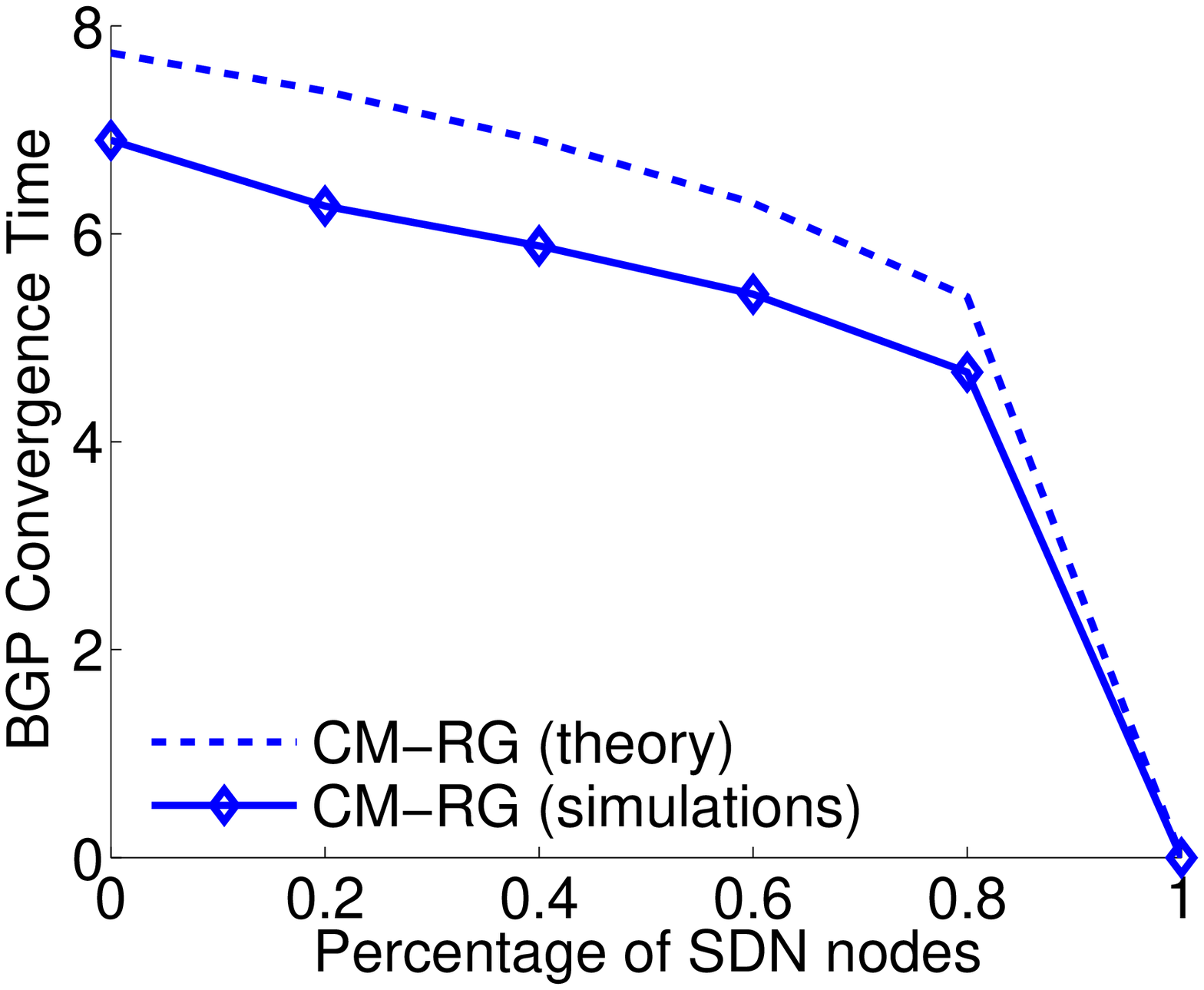}\label{fig:TvsPenetration-cm}}
\caption{BGP convergence time vs. SDN penetration (i.e, $\frac{k}{N}$) in networks with $N=300$. (a) Full-mesh and Poisson graph ($p = \frac{1}{60}$) topologies; and (b) CM-RG topology with degree sequence $d_{i}$, $d_{i}\in[5,200]$ and power-law distributed with exponent $2$.}
\label{fig:TvsPenetration}
\end{figure}
\section{Case Study: BGP Convergence at the Internet Core Network}\label{sec:case-study}
In this section, we use our results to study the effects of routing centralization in the core of the Internet. We show that our analysis can provide useful insights related to the evolution of the Internet, and answer questions like ``\textit{Internet flattening} or \textit{routing centralization} can improve more the BGP convergence?''. 

We first need to model the core of the Internet. We consider the following setting, which captures the structure in the core of the Internet in an (a) abstract (and thus analytically tractable), (b) generic (and thus easily adapted to network characteristics and changes), and (c) realistic (since it takes into account routing policies, e.g., customer to provider, peering links, valley-free model) way.

Specifically, we assume that the Internet core network is composed of $N_{1}$ tier-1 ISPs, and $N_{2}$ tier-2 ISPs. Let tier-1 ISPs to peer with each other with probability $p_{11}$ (Poisson graph), while tier-2 ISP peer with probability $p_{22}$. Also, tier-1 ISPs offer transit to tier-2 ISPs; each tier-2 ISP might have one or more transit providers. Let $p_{12}$ be the probability that a tier-1 and a tier-2 ISPs are connected.

Consider a prefix announced by AS-x, a tier-2 ISP. Due to the routing policies (peering/transit), the BGP updates are forwarded as follows. The tier-2 \textit{peers} of AS-x receive the BGP updates directly from AS-x. The same holds also for the tier-1 ISPs that are connected to AS-x. The remaining tier-1 ISPs receive the updates from other tier-1 ISPs, and the remaining tier-2 ISPs from their tier-1 providers. Fig.~\ref{fig:case-study} demonstrates such an example topology and the respective eligible paths of BGP updates.

Two recent trends, related to the evolution of the Internet, have been recently observed and/or proposed: (a) Internet flattening~\cite{Gregori-Impact-IXPs-ComCom-2011}, i.e., more peer connections appear. In the above model, higher values of $p_{22}$ correspond to more flattening. (b) Inter-domain SDN, where some ISPs centralize their routing. Network economics and feasibility studies indicate that routing centralization is more probable to start from larger ISPs~\cite{Kotronis-CXP-SOSR-2016}; hence, we assume that $k_{1}\in[1,N_{1}]$ tier-1 ISPs belong to an SDN cluster.

Using the results of the previous sections, and making similar analytic arguments, in the remainder, we derive expressions for the BGP convergence time as a function of the network parameters ($N_{1},k_{1},N_{2},p_{11},p_{22},p_{12}$). Therefore, it becomes straightforward to compare the two approaches, Internet flattening vs. routing centralization, by evaluating the expressions for different values of $p_{22}$ and $k_{1}$.

\begin{figure}
\includegraphics[width=\linewidth]{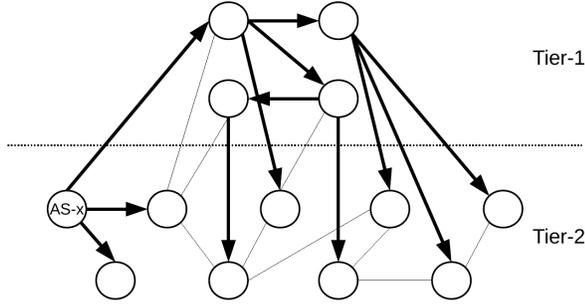}
\caption{Internet core topology example. Lines/arrows denote connections between ASes. Arrows denote the eligible paths of BGP updates.}
\label{fig:case-study}
\end{figure}

Let us use the notation 
\begin{equation}
G_{fm}(N,k,\lambda)~~~\text{and}~~~G_{p}(N,k,p,\lambda)
\end{equation}
for full-mesh and Poisson graph networks (with the given parameters), respectively.

The BGP convergence delay is given by the following Result.

\begin{result}
The BGP convergence delay in the core of the Internet is given by
\begin{equation}
E[T] = max\left\{E[T_{peering}], E[T_{transit}]\right\}
\end{equation}
where
\begin{align}
E[T_{peering}] 	&= E[T|G_{fm}\left(N_{2}\cdot p_{22},1, \lambda\right)] \\
E[T_{transit}]  &= E[T_{x,tier-1}] + E[T_{tier-1}] + E[T_{tier-1,tier-2}]
\end{align}
and
\begin{align}
E[T_{x,tier-1}] &	= \frac{1}{p_{12}\cdot N_{1}} \\
E[T_{tier-1}]	&\approx E[T|G_{p}\left(N_{1},k_{1},p_{11},\lambda\right)] \\
E[T_{tier-1,tier-2}] &\approx  E[T|G_{fm}\left(N_{2}\cdot (1-p_{22}),1, N_{1}\cdot p_{12}\cdot \lambda\right)]
\end{align}
\end{result}
\begin{proof}
$E[T_{peering}]$ is the time till every peer of AS-x (i.e., direct connections to tier-2 ISPs) to receive the updates, and $E[T_{transit}]$ the respective time for the non-peering ASes (which get the updates through transit, i.e., the tier-1 ISPs).

AS-x has (on average) $N_{2}\cdot p_{22}$ peers, which are directly connected to it. Hence, for the BGP update dissemination process we can use the full-mesh network model with $k=1$ (i.e., without centralization), since we consider centralization only for tier-1 ISPs. As a result, it holds that
\begin{equation}
E[T_{peering}] 	= E[T|G_{fm}\left(N_{2}\cdot p_{22},1, \lambda\right)]
\end{equation}

To calculate $E[T_{transit}]$, we split it into three parts, i.e., 
\begin{equation}
E[T_{transit}]  	= E[T_{x,tier-1}] + E[T_{tier-1}] + E[T_{tier-1,tier-2}]
\end{equation}
where (i) $T_{x,tier-1}$ is the time till the first tier-1 AS (a transit provider of AS-x) receives the BGP update, (ii) $T_{tier-1}$ is the time, after $T_{x,tier-1}$, needed for every tier-1 ISP to get the BGP update, and (iii) $T_{tier-1,tier-2}$ is the time, after $T_{tier-1}$, needed for the tier-1 ASes to send the BGP updates to their customers (i.e., tier-2 ASes).

Using the Markovian properties of BGP update times and making similar arguments as before, we can show that
\begin{equation}
E[T_{x,tier-1}] 	= \frac{1}{p_{12}\cdot N_{1}}
\end{equation}
and
\begin{equation}
E[T_{tier-1}]	\approx E[T|G_{p}\left(N_{1},k_{1},p_{11},\lambda\right)]
\end{equation}
Also, it holds that
\begin{equation}
E[T_{tier-1,tier-2}] \approx  E[T|G_{fm}\left(N_{2}\cdot (1-p_{22}),1, N_{1}\cdot p_{12}\cdot \lambda\right)]
\end{equation}
since there exist $N_{2}\cdot (1-p_{22})$ tier-2 ASes that are not peers with AS-x (and thus will receive the update from tier-1 ASes), and each of them is (on average) connected to $N_{1}\cdot p_{12}$ tier-1 ISPs (and can receive updates from any of them).
\end{proof}
%\begin{align}
%E[T_{x,tier-1}] 	&= \frac{1}{p_{12}\cdot N_{1}}\\
%E[T_{tier-1}]	&= E[T|G_{p}\left(N_{1},k_{1},p_{11},\lambda\right)]\\
%E[T_{tier-1,tier-2}]				&=  E[T|G_{fm}\left(N_{2}\cdot (1-p_{22}), N_{1}\cdot p_{12}\cdot \lambda\right)]
%\end{align}

\bibliographystyle{ieeetr}
%\bibliographystyle{acm}
%\small

\end{document}